\documentclass[runningheads]{llncs}
\usepackage{graphicx}
\usepackage{xspace}
\usepackage{tikz}
\usetikzlibrary{shapes, petri, calc, arrows}
\usepackage{extarrows}
\usepackage{paralist}
\usepackage{enumitem}
\usepackage{amsmath, amssymb}
\usepackage{algorithm}
\usepackage[noend]{algpseudocode}
\usepackage{pgfplots} 
\pgfplotsset{compat = newest}
\usepackage{wrapfig}
\usepackage[hidelinks]{hyperref}
\usepackage{thmtools, thm-restate}

\usepackage{todonotes}

\newcommand{\m}[1]{\mathsf{#1}}
\newcommand{\mc}[1]{\mathcal{#1}}
\renewcommand{\vec}[1]{\overline{#1}}
\renewcommand{\phi}{\varphi}

\renewcommand{\leq}{\leqslant}
\renewcommand{\geq}{\geqslant}

\renewcommand{\AA}{\mc A} % action set of DDS
\newcommand{\BB}{\mc B} % DDS
\newcommand{\CC}{\mc C} % constraint set
\newcommand{\NN}{\mc N} % DPN
\newcommand{\inn}{\,{\in}\,} % narrow \in 
\newcommand{\eqn}{\,{=}\,} % narrow =
\newcommand{\leqn}{\,{\leq}\,} % narrow <=
\newcommand{\geqn}{\,{\geq}\,} % narrow >=

\newcommand{\constraint}{c}

\newcommand{\Dom}{\mathcal D} % domain of a substitution
\newcommand{\binit}{b_{\mathit{I}}} % initial control state of DDS
\newcommand{\alphainit}{\alpha_{\mathit{I}}} % initial assignment of DDS
\newcommand{\goto}[1]{\mathrel{\raisebox{-2pt}{$\xrightarrow{#1}$}}}

\newcommand{\CG}[1][\BB]{\textup{CG}_{#1}} % constraint graph
\newcommand{\CGof}[1]{\textup{CG}_\BB(#1)} % constraint graph
\newcommand{\alpharen}{\widehat \alpha}
\newcommand{\update}{\mathit{update}} % update operation in CG construction
\newcommand{\guard}{\mathit{guard}} % guard function
\newcommand{\vwrite}{\mathit{write}} % set of written variables
\newcommand{\vread}{\mathit{read}} % set of read variables
\newcommand{\qe}{\mathit{qe}} % quantifier elimination
\newcommand{\trans}[1]{\Delta_{#1}} % transition formula for action
\newcommand{\foo}{\tikz{\draw (0,0) -- (.3,0); \node[fill, minimum size=3pt, circle, anchor=center, inner sep=0pt] at (.15,0) {};}}

\newcommand{\domino}[4]{\text{\tikz[baseline=-0.5ex]{\node[scale=.6, inner sep=0pt]{$%
\left[\begin{array}{@{\,}l@{=}l@{\,}}{#1}&{#2}\\{#3}&{#4}\\\end{array}\right]$}}}}

\newcommand{\bool}{\mathtt{bool}}
\newcommand{\integer}{\mathtt{int}}

\newcommand{\rational}{\mathtt{rat}}

\newcommand{\tool}{\texttt{ada}\xspace}

\tikzstyle{place}=[draw, circle, inner sep=1.5pt, line width=.7pt, scale=.6, minimum width=5mm]
\tikzstyle{trans}=[draw, rectangle, inner sep=1.5pt, line width=.7pt, scale=.6, minimum width=4mm, minimum height=4mm]
\tikzstyle{state}=[draw, circle, inner sep=1.5pt, line width=.7pt, scale=.6]
\tikzstyle{edge}=[draw, ->, line width=.5pt]
\tikzstyle{action}=[scale=.6]
\tikzstyle{caption}=[scale=.9]
\newcommand{\cgnode}[2]{{#1}\nodepart{two}{#2}}

\newcommand{\lemref}[1]{Lem.~\ref{lem:#1}}

\newcommand{\defref}[1]{Def.~\ref{def:#1}}

\newcommand{\secref}[1]{Sec.~\ref{sec:#1}}
\newcommand{\tabref}[1]{Tab.~\ref{tab:#1}}

\newcommand{\exaref}[1]{Ex.~\ref{exa:#1}}
\newcommand{\figref}[1]{Fig.~\ref{fig:#1}}

\renewcommand{\algref}[1]{Alg.~\ref{alg:#1}}
\renewcommand{\eqref}[1]{(\ref{eq:#1})}

\begin{document}
\title{Soundness of Data-Aware Processes with Arithmetic Conditions
}
\author{Paolo Felli \and Marco Montali \and Sarah Winkler%
\thanks{This work is partially supported by the UNIBZ projects DaCoMan, QUEST, SMART-APP, VERBA, and WineId.}}
\authorrunning{Felli, Montali, Winkler}
\institute{Free University of Bozen-Bolzano
\email{\{pfelli,montali,winkler\}@inf.unibz.it}}
\maketitle              % typeset the header of the contribution
\begin{abstract}

Data-aware processes represent and integrate structural and behavioural constraints in a single model, and are thus increasingly investigated in business process management and information systems engineering. In this spectrum, Data Petri nets (DPNs) have gained increasing popularity thanks to their ability to balance simplicity with expressiveness.
%: they capture sophisticated processes equipped with case variables and data conditions, in turn capturing decisions and constrained updates. 
The interplay of data and control-flow makes checking the correctness of such models, specifically the well-known property of soundness, crucial and challenging.  
A major shortcoming of previous approaches for checking soundness of DPNs is that they consider data conditions without arithmetic, an essential feature when dealing with real-world, concrete applications. 
In this paper, we attack this open problem by providing a foundational and operational framework for assessing soundness of DPNs enriched with arithmetic data conditions. 
The framework comes with a proof-of-concept implementation that, instead of relying on ad-hoc techniques, employs off-the-shelf established SMT technologies. 
The implementation is validated on a collection of examples from the literature, and on synthetic variants constructed from such examples.

\end{abstract}
\section{Introduction}
Integrating structural and behavioral aspects to holistically capture how information systems dynamically operate over data through actions and processes is a central problem in business process management (BPM) \cite{Reichert12} and information systems engineering \cite{Snoe14}. This is witnessed by the mutual cross-fertilization of the two areas on this topic, with models and approaches originating from BPM and its underlying formal foundations %and
being % sarah
then applied to information and enterprise systems \cite{PWOB19,FetR20,RRMR21}, and vice-versa \cite{ACMA19,SnSW21}. 

The interplay of data and control-flow makes checking the correctness of such models crucial and challenging. From the formal point of view, the problem is undecidable even for severely restricted models and correctness properties, both in the case of simple data variables \cite{LFM18} and richer relational structures \cite{CGM13,DHLV18}.  
From the modeling perspective, the difficulty in combining these two dimensions is exacerbated by the fact that, more and more, models are obtained through a two-step approach: a first, automated discovery step produces a baseline model from event data, followed by a refinement and modification step driven by human ingenuity. 
 The following example illustrates the challenge.

\begin{figure}[t] 
\resizebox{\columnwidth}{!}{
\centering
\begin{tikzpicture}[node distance=49mm,>=stealth']
\node[state] (p1) {$\m p_1$}; % pl1
\node[state, right of=p1, xshift=-13mm] (p2) {$\m p_2$}; % pl12
\node[state, right of=p2, xshift=-3mm] (p3) {$\m p_3$}; % pl6
\node[state, right of=p3, xshift=-7mm] (p4) {$\m p_4$}; % pl7
\node[state, below of=p2, yshift=9mm, double] (end) {$\m {end}$}; 
\node[state, right of=p4, xshift=-3mm] (p5) {$\m p_5$}; % pl10
\node[state, below of=p5,yshift=9mm] (p6) {$\m p_6$}; % pl13
\node[state, left of=p6] (p7) {$\m p_7$}; % pl14
\node[state, left of=p7] (p8) {$\m p_8$}; % pl15
\draw[edge] (p1) -- 
  node[above, action] {\textsf{create fine}}
  node[below, action] {$a^w, t^w, d^w, p^w \geqn 0$}
  (p2);
\draw[edge] (p2) to[loop above, looseness=9]
  node[above, yshift=2.5mm,action] {\textsf{payment}}
  node[above, action] {$t^w \geqn 0$}
  (p2);
\draw[edge] (p2) -- 
  node[above, action] {\textsf{send fine}}
  node[below, action] {$0 \leqn \mathit{ds}^w \leqn 2160 \wedge e^w \geqn 0$}
  (p3);
\draw[edge] (p2) -- 
  node[left, action, near end] {\textsf{$\tau_1$}}
  node[left, action, near end, yshift=-4mm] {$d\,{\neq}\,0 \vee (p \eqn 0 \wedge t^r \geqn a^r)$}
  (end);
\draw[edge] (p3) to[loop above, looseness=9]
  node[above, yshift=2.5mm,action] {\textsf{payment}}
  node[above, action] {$t^w \geqn 0$}
  (p3);
\draw[edge] (p3) -- 
  node[above, action] {\textsf{insert notification}}
  (p4);
\draw[edge, rounded corners] (p3) -- ($(p3) + (-.7,-.7)$) --
  node[left, action, above] {\textsf{$\tau_2$}}
  node[left, action, below] {$t^r \geqn a^r\,{+}\,e^r$}
  ($(p3) + (-2.5,-.7)$) --(end);
\draw[edge] (p4) to[loop above, looseness=11, out=170, in=130]
  node[above, yshift=2.5mm,action] {\textsf{payment}}
  node[above, action] {$t^w \geqn 0$}
  (p4);
\draw[edge] (p4) to[loop right, looseness=11, out=110, in=70]
  node[right, yshift=2.5mm,action] {\textsf{add penalty}}
  node[right, xshift=.5mm, action] {$a^w \geqn 0$}
  (p4);
\draw[edge, rounded corners] (p4) -- ($(p4) + (.2,-.3)$) --
  node[above, action] {\textsf{appeal to judge}}
  node[below, action] {$0 \leqn \mathit{dj}^w \leqn 1440 \wedge d^w \geqn 0$}
  ($(p5) + (-.2,-.3)$) -- (p5);
\draw[edge, rounded corners] (p4) -- ($(p4) + (-.5,-1.2)$) --
  node[above, action] {\textsf{credit collection}}
  node[below, action] {$t^r\,{<}\,a^r\,{+}\,e^r$}
  ($(p4) + (-2.6,-1.2)$) -- (end);
\draw[edge, rounded corners] (p4) -- ($(p4) + (-.8,-.5)$) --
  node[above, action] {\textsf{$\tau_3$}}
  node[below, action] {$t^r \geqn a^r\,{+}\,e^r$}
  ($(p4) + (-2.6,-.5)$) --(end);
\draw[edge, rounded corners] (p5) -- ($(p5) + (-.2,.3)$) --
  node[above, action] {\textsf{$\tau_5$}}
  node[below, action] {$d^r \eqn 0$}
  ($(p4) + (.2,.3)$) -- (p4);
\draw[edge, rounded corners] (p4) -- ($(p4) + (.4,-1.3)$) --
  node[above, action] {\textsf{appeal to prefecture}}
  node[below, action] {$0 \leqn \mathit{dp}^w \leqn 1440$}
  ($(p4) + (2.7,-1.3)$) -- (p6);
\draw[edge] (p6) to
  node[left, action, above] {\textsf{send to prefecture}}
  node[left, action, below] {$d^w \geqn 0$}
  (p7);
\draw[edge] (p7) to
  node[left, action, above] {\textsf{result prefecture}}
  node[left, action, below] {$d^r \eqn 0$}
  (p8);
\draw[edge, rounded corners] (p7) -- ($(p7) + (-.2,-.3)$) --
  node[left, action, above, near end] {\textsf{$\tau_6$}}
  node[left, action, below, near end] {$d^r \eqn 1$}
  ($(end) + (.4,-.3)$) -- (end);
\draw[edge, rounded corners] (p5) -- ($(p5) + (.3,0)$) -- ($(p5) + (.3,-3.1)$) --
  node[left, action, above] {\textsf{$\tau_4$}}
  node[left, action, below] {$d^r \eqn 2$}
  ($(p5) + (-7.7,-3.1)$)  -- (end);
\draw[edge, rounded corners] (p8) -- ($(p8) + (.5,.5)$) --
  node[left, action, above] {\textsf{notify}}
  ($(p8) + (2.5, .5)$) -- (p4);
\end{tikzpicture}
}
\caption{Data-aware process for road fines \cite{MannhardtLRA16}.}
\label{fig:fine}
\end{figure}

\begin{example}
\label{exa:road fines}
A management process for road fines from an information system of the Italian police was presented as in~\cite{MannhardtLRA16} using a Data Petri nets (DPN). DPNs have gained increasing popularity thanks to their ability to balance simplicity with expressiveness. They focus on the evolution of a single (case) object evolved by the process (or a fixed number of inter-related objects), combining a Petri net-based control-flow with case variables and data conditions, capturing decisions and constrained updates. The process maintains seven case data variables:
$a$ (amount), $t$ (total amount), $d$ (a dismissal code), $p$ (points deducted),
$e$ (expenses), and three time intervals $\mathit{ds}$, $\mathit{dp}$, $\mathit{dj}$. %\todo{briefly explain ex}
The process starts by creating a fine for a traffic offense in the system (\textsf{create fine}).
A notification is sent to the offender within 90 days, i.e., 2160h, by action \textsf{send fine}) and this is entered in the system (\textsf{insert notification}).
If the offender pays an amount $t$ that exceeds the fine $a$ plus expenses $e$, the process terminates via $\tau_1$, $\tau_2$, or $\tau_3$.
For the less happy paths, there is a \textsf{credit collection} action if the paid sum was not enough; and the offender may file a protest, via
\textsf{appeal to judge}, \textsf{appeal to prefecture}, and subsequent actions. The appeals again need to respect a certain time frame.

For simplicity, in Figure~\ref{fig:fine} we present the model  as a transition system instead of a Petri net.
It was generated from real-life logs through multi-perspective process mining techniques, then enriched manually with more sophisticated arithmetic constraints extracted from domain knowledge \cite{MannhardtLRA16}. \emph{What is not obvious is that the process gets stuck in state
$\m p_7$ if \textsf{send to prefecture} writes value $d\,{>}\,1$.}
\end{example}

Examples like this call for a virtuous circle where process mining, human modelling, and \emph{automated} verification techniques for correctness checking empower each other. 
A well-established formal notion of correctness for dynamic systems is that of \emph{soundness} \cite{Aalst1998workflow}, defined over the well-known Petri net class of workflow nets. %\todo{Marco: I removed the intuition of soundness. Is it needed later?} 
Intuitively, this property requires
\begin{inparaenum}[(i)]
\item that there are no activities in the process that cannot be executed in any of the possible executions; 
\item that from every reachable configuration the process can always be concluded by reaching a \emph{final} configuration and 
\item that final configurations are always reached in a `clean way', without leaving any thread of the process still hanging. 
\end{inparaenum}
After the seminal work in \cite{Aalst1998workflow}, which solely focuses on the evolution of single process instances in pure control-flow terms, several follow-up approaches were brought forward to define and study soundness 
% considering 
for % sarah shorten
richer control-flow structures \cite{AHHS11}, several isolated cases \cite{HeSV04}, and presence of resources \cite{SidS13}, showing decidability of the problem without entering into the engineering of verification tools.

When considering data-aware processes, the standard formulation of soundness is insufficient, as it does not consider how data affects the execution. This makes prior works not readily applicable to solve the problem. Refined notions of soundness have in fact been put forward to take data into account. 
%\cite{BatoulisHW17,LFM18} to take data into account \cite{deLeoniFM21,deLeoniFM21Jods}.
 Specifically, in~\cite{LFM18} the property of \emph{data-aware} soundness was obtained by lifting the standard soundness property of workflow nets to DPNs \cite{Mannhardt18,LFM18} (see the example above), by resorting to a translation to colored Petri nets. 
However, data conditions attached to activities were restricted to variable-to-constant comparisons. The approach was later extended to DPNs with a guard language that supports direct comparison of case data \cite{deLeoniFM21}.
In parallel, \cite{BatoulisHW17} introduced notions of \emph{decision-aware} soundness, where the focus is on data consumed and produced by (DMN) decision tables attached to the process. It was later shown in \cite{deLeoniFM21Jods} how DPNs could be used to capture BPMN processes enriched with DMN S-FEEL decision tables, and how the different decision-aware soundness notions \cite{BatoulisHW17} could be recast as data-aware soundness \cite{LFM18}. 
%of of to recast decision-aware In that work, a technique for verifying the soundness (and beyond) of DBPMN models is shown, which is based on a translation from DBPMN processes to DPNs, so that the very same techniques of \cite{deLeoniFM21} can be employed also in this setting. 

While data-aware soundness is a crucial notion that captures also the problem in \exaref{road fines},
a common shortcoming present in the literature is the limited expressivity of data conditions attached to activities and decision rules: \emph{they cannot handle expressions with arithmetic computations}.  
For instance, one can check that the current credit card balance $b$ is equal or larger than the price $p$ of the purchased item (i.e., that $b\geq p$), but not that it is greater than the price plus some threshold amount $t$ that could be obtained through a human task (i.e., that $b\geq p + t$). 
Clearly, this makes the existing technique not applicable to a very large number of real world applications (for instance,  \exaref{road fines}), revealing a research gap in the field that motivates the need of novel results in this spectrum.

\noindent
\textbf{Contributions and methodology.} 
Having identified this open research problem, we aim at contributing to the advancement of the body of knowledge in information systems engineering by answering three research questions:
\begin{compactenum}
\item Is soundness checking decidable for DPNs equipped with arithmetic?
\item Is there an operational way to conduct the check?
\item Is this operational way effective from the computational point of view?
\end{compactenum}
We answer these through theoretical and algorithmic research, and through the creation of a concrete IT proof-of-concept artifact for soundness checking.

Specifically, we focus on DPNs supporting unlimited addition of variables but only constant multiplication, that is, \emph{linear arithmetic}, which captures many real-world use cases. We address the first two research questions at once by  
lifting the approach in \cite{deLeoniFM21} to our richer setting, introducing a soundness checking procedure consisting of three algorithmic steps: 
\begin{inparaenum}[(1)]
\item we transform the DPN into a labelled transition system called \emph{data-aware dynamic system} (DDS) \cite{LFM20};
\item we %abstract the possibly infinite state-space by means of 
construct
a \emph{constraint graph}, which acts as a symbolic representation of the reachable state space via a finite set of formulas;
\item a set of satisfiability checks is performed using the formulas in the graph, and we prove that the DPN is unsound if and only if one of these checks succeeds.
\end{inparaenum}
The constraint graph built for a DDS with arithmetic may in general be infinite. However, it is finite and computable, so that our check becomes a decision procedure, when the given process guarantees that reachable configurations
are suitably limited (e.g. in that only a bounded part of the computation
history is relevant, or the constraint language is sufficiently restricted).
This requirement holds for well-identified classes of processes,
formally captured by a \emph{finite history set} \cite{ada}. For instance, it applies to all DPNs used in our evaluation, including \exaref{road fines}.

Towards answering the third research question, we provide a proof-of-concept implementation of our framework in the tool \tool. Being research in this setting at an early stage, we cannot rely on well-established empirical or experimental methods to validate this IT artifact. To mitigate this problem, we proceed as follows. First and foremost, instead of relying on ad-hoc techniques, our tool employs off-the-shelf SMT solvers as a backend. This guarantees that the main computation burden, namely the satisfiability checks in the third algorithmic step, is handled by third-party, industrially-validated software. Secondly, since there is no benchmark for DPNs, we set up a preliminary, performance evaluation in two steps:
\begin{inparaenum}[\it (i)]
\item we collect, and check soundness of, all DPN examples/case studies present in the literature to model real-world data-aware processes in information systems of various types;
\item we construct synthetic variants of some of these examples, in order to test how the performance of \tool changes by increasing actions, variables and conditions present in the model.  
\end{inparaenum}
%To illustrate feasibility, we report on the application of \tool to a range of DPNs from the literature that model real-world business processes . These experiments confirmed that our approach is effectively and efficiently applicable in practice to detect soundness violations. 
%Indeed, we exhibited a soundness violation in at least one case.

%\medskip
The paper is structured as follows. In \secref{background}, we fix our DPN model and define data-aware soundness, illustrating its high-level verification procedure in \secref{overview}. The following sections %show how this can be done in practice: 
detail the required steps:
in \secref{dds} we relate data-aware soundness of a DPN to that of a corresponding transition system. We explain the constraint graph in \secref{cg}, and show in \secref{soundness} how it can be used to check data-aware soundness. Our implementation and experiments are the topic of \secref{implementation}. 
In \secref{conclusion} we conclude and comment on future work.

\section{Background}
\label{sec:background}

In this section we summarize some background on constraints, DPNs and data-aware dynamic systems, as well as data-aware soundness.
%\smallskip

\newcommand{\sort}{type}

\smallskip
\noindent \textbf{Constraints.}
We start by fixing a set of data types for the variables manipulated by a process:
let $\Sigma = \{\texttt{bool}, \texttt{int}, \texttt{rat}\}$ with associated domains 
of booleans $\Dom(\bool) = \mathbb B$,
integers $\Dom(\integer) = \mathbb Z$, and
rationals $\Dom(\rational) = \mathbb Q$.
We assume a fixed set of \emph{process variables} $V$,
so there is a function $\sort\colon V \mapsto \Sigma$ assigning a type to each variable. 
For instance, in \exaref{road fines} the set of process variables is
$V = \{a, d, dj, dp, ds, p, t\}$ all of type $\integer$ (i.e., $\sort(a)=\integer$, etc).
For a type $\sigma\inn \Sigma$, 
$V_{\sigma}$ denotes the subset of variables of type $\sigma$. 
To manipulate variables, we consider expressions $\constraint$ with the following grammar:
\begin{align*}
 c &:= x_{\bool} \mid b \mid n_1~op~n_2 \mid r_1~op~r_2 
  \mid c_1 \wedge c_2  \\
  op := {\neq} \mid {=} \mid {\geq} \mid {>} \quad
 n &:= x_{\integer} \mid k \mid k_1 \cdot n_1 + k_2 \cdot n_2 \quad
 r = x_{\rational} \mid q \mid q_1 \cdot r_1 + q_2 \cdot r_2
\end{align*}
where: $x_{\bool} \in V_{\bool}, x_{\integer} \in V_{\integer}$, and $x_{\rational} \in V_{\rational}$ respectively denote a boolean, integer, and rational variable, while $b \in \mathbb B$, $k \in \mathbb Z$, and $q \in \mathbb Q$ respectively denote a boolean, integer, and rational constant. We consider booleans, integers, and rationals as three prototypical examples of three datatypes, respectively relying on a finite, infinite discrete, and infinite dense domain. Similar datatypes, such as strings equipped with equality and real numbers, can be seamlessly handled.

\noindent
% Standard equivalences apply, hence %disjunction (i.e., $\lor$) and 
% comparisons $\neq$, $<$, $\leq$ can be used as well.
These expressions will be used to capture conditions on the values of variables that are read and written during the execution of process activities.
For this reason, we call them \emph{constraints}. The set of constraints over $V$ is denoted $\mathcal C(V)$.

For our process variables $V$, we consider two disjoint sets of \emph{annotated} variables $V^r = \{v^r \mid v\inn V\}$ and $V^w = \{v^w \mid v\inn V\}$ which are read and written by process activities, respectively, as explained below, and we assume $\sort(v^r) = \sort(v^w) = \sort(v)$ for every $v\in V$. 
% We denote the set of all constraints over $V^r \cup V^w$ by $\mathcal C(V)$.
%
For instance, the constraint $t^r \geq a^r + e^r$ in \exaref{road fines}
dictates that the current value of variable $t$ is greater or equal than the
sum of the values of $a$ and $r$; whereas $0 \leq dj^w \wedge dj^w\leq 1440$
requires that the new value given to $dj$ (i.e., assigned to $dj$ as a result of the execution of the activity to which this constraint is attached) is 
between $0$ and $1440$. On the other hand, $a^w > a^r$ would mean that the
new value of $a$ is larger than its current value.
More generally, given a constraint $c$ as above, we refer to the annotated variables in $V^r$ and $V^w$ that appear in $\constraint$ as the \emph{read} and \emph{written variables}, respectively. 

An \emph{assignment} $\alpha$ is a total function
$\alpha \colon V \mapsto D$ mapping each variable in $V$ to a value
in its domain.
We say that $\alpha$ \emph{satisfies} a constraint $c$ over $V$, 
written $\alpha \models c$, if the evaluation of $c$ under $\alpha$ is true.
For instance, the assignment $\alpha$ such that $\alpha(t) = 10$, $\alpha(a)=7$,
and $\alpha(v)=0$ for $v\in V$ otherwise, satisfies $t^r \geq a^r + e^r$.

Our constraint language is that of linear arithmetic over integers and rationals, which is decidable, and for which a range of mature SMT (satisfiability modulo theories) solvers~\cite{Z3,DR0BT14} is available. %to check satisfiability of these formulas.
Moreover, linear arithmetic is known to enjoy \emph{quantifier elimination}~\cite{Presburger29}:
if $\phi$ is a formula with atoms in $\CC(V \cup \{x\})$,
there is some $\phi'$ with free
vari\-ables $V$ that is logically equivalent to $\exists x. \phi$, 
i.e., $\phi'\,{\equiv}\,\exists x. \phi$.
We assume that $\qe$ is a quantifier elimination procedure that
returns such a formula, as implemented in off-the-shelf SMT solvers.
%\smallskip

\noindent
%\textbf{Data Petri Nets.}
%We adopt the following standard definition of  DPNs \cite{Mannhardt18,MannhardtLRA16}. 
We adopt the following standard definition of Data Petri Nets (DPNs) \cite{Mannhardt18,MannhardtLRA16}. 

\begin{definition}[DPN]
A %\emph{Petri net with data} (DPN) 
DPN is a tuple
$\NN = \langle P, T, F, \ell, \AA, V,\guard\rangle$, where
\begin{inparaenum}[\it (1)]
\item $\langle P, T, F, \ell\rangle$ is a Petri net with non-empty, disjoint sets of places $P$ and transitions $T$, 
a flow relation $F:(P \times T)\cup(T \times P)\mapsto\mathbb{N}$ and a labelling function $\ell\colon T\mapsto \AA$, where $\AA$ is a finite set of activity labels; 
\item $V$ is a set of process variables (all with a type); and
\item $guard\colon T \mapsto \CC(V^r \cup V^w)$ is a guard mapping. 
\end{inparaenum} 
\end{definition}

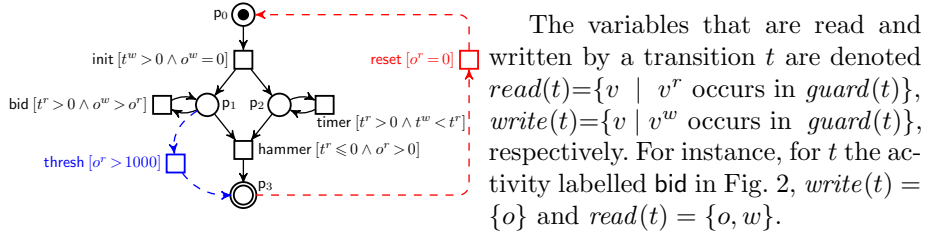
\begin{wrapfigure}[9]{l}{6.1cm}
\vspace*{-1cm}
\centering
\begin{tikzpicture}[node distance=10mm,>=stealth', xscale=1.2]
% \begin{scope}
\node[place, tokens = 1] (p0) {};
\node[trans, below of=p0] (init) {};
\node[place, below of=init, xshift=-8mm] (p1) {};
\node[place, below of=init, xshift=8mm] (p2) {};
\node[trans, left of=p1, yshift=0mm] (bid) {};
\node[trans, right of=p2, yshift=0mm] (dec) {};
\node[trans, below of=p1, xshift=8mm] (exp) {};
\node[place, below of=exp,double] (end) {};
\node[trans, right of=init, xshift=40mm, red] (reset) {};
\node[trans, below of=bid, yshift=-3mm, xshift=3mm, blue] (thresh) {};
\draw[edge] (p0) -- (init);
\draw[edge] (init) -- (p1);
\draw[edge] (init) -- (p2);
\draw[edge] (p1) to[bend left] (bid);
\draw[edge] (bid) to[bend left] (p1);
\draw[edge] (p2) to[bend left] (dec);
\draw[edge] (dec) to[bend left] (p2);
\draw[edge] (p1) to (exp);
\draw[edge] (p2) to (exp);
\draw[edge] (exp) to (end);
\draw[edge, rounded corners, red, dashed] (end) -- (end -| reset) -- (reset);
\draw[edge, rounded corners, red, dashed] (reset) -- (reset |- p0) -- (p0);
\node[action,anchor=east,xshift=-2mm] at (init) {$\mathsf{init}\:[t^w\,{>}\,0 \wedge o^w\,{=}\,0]$};
\node[action,anchor=east,xshift=-2mm] at (bid) {$\mathsf{bid}\:[t^r\,{>}\,0 \wedge o^w\,{>}\,o^r]$};
\node[action,anchor=west,xshift=2mm] at (exp) {$\mathsf{hammer}\:[t^r\,{\leq}\,0 \wedge o^r\,{>}\,0]$};
\node[action,anchor=north west,xshift=-3mm,yshift=-1mm] at (dec) {$\mathsf{timer}\:[t^r\,{>}\,0 \wedge t^w\,{<}\,t^r]$};
\node[action,anchor=east,xshift=-2mm, red] at (reset) {$\mathsf{reset}\:[o^r\,{=}\,0]$};
\node[action,anchor=east,xshift=-2mm,blue] at (thresh) {$\mathsf{thresh}\:[o^r\,{>}\,1000]$};
\node[action,anchor=east,xshift=-2mm] at (p0) {$\m p_0$};
\node[action,anchor=west,xshift=2mm] at (p1) {$\m p_1$};
\node[action,anchor=east,xshift=-2mm] at (p2) {$\m p_2$};
\node[action,anchor=west,xshift=2mm,yshift=2mm] at (end) {$\m p_3$};
\draw[edge, dashed,blue] (p1) to[bend right] (thresh);
\draw[edge,dashed,blue] (thresh) to[bend right] (end);
\end{tikzpicture}
\caption{DPN for simple auction model.}
\label{fig:dpn}
\end{wrapfigure}

\begin{example}
\label{exa:auction}
%As a running example, we will 
Consider a simple auction process modeled by the DPN in \figref{dpn}. The initial and final markings are $M_I = \{\m p_0\}$ and $M_F = \{\m p_3\}$. It maintains the set of variables $V = \{o,t\}$, where $o$ (domain $\mathbb Q$) holds the last offer issued by a bidder, and $t$ (domain $\mathbb Z$) is a timer. The initial assignment is $\alphainit(o) = \alphainit(t) = 0$.
We briefly explain the working of the process: 
the action \textsf{init} initializes the timer $t$ to a positive value (e.g., of days) 
and the offer $o$ to $0$;
as long as the timer has not expired, it can be decreased (action \textsf{timer}), or
bids can be issued, increasing the current offer (\textsf{bid});
the item can be sold if the timer expired and the offer is positive (\textsf{hammer}).
We denote this DPN, consisting of all actions drawn in black in \figref{dpn}, by $\NN$.
For illustration purposes, we will also consider two variants of this DPN:
$\NN_{\m{reset}}$ extends $\NN$ by a \textsf{reset} action that restarts the process if the offer in the final state is 0 (drawn in red), and
$\NN_{\m{thresh}}$ adds to $\NN$ the
transition \textsf{thresh} which leads to the final state if the offer exceeds a threshold (drawn in blue). 
\end{example}

%As customary, given $x\in P \cup T$, we use $\pre{x}:=\set{y\mid F(y,x)>0}$ to denote the \emph{preset} of $x$ and $\post{x}:=\set{y\mid F(x,y)>0}$ for the \emph{postset} of $x$. 
The variables that are read and written by a transition $t$ are denoted $read(t){=}\{v \mid v^r \text{ occurs in }\guard(t) \}$,
$\vwrite(t) {=} \{v \mid v^w \text{ occurs in }$ $\guard(t) \}$,
respectively.
For instance, for $t$ the activity labelled $\mathsf{bid}$ in \figref{dpn},
$\vwrite(t) = \{o\}$ and $\vread(t) = \{o,w\}$.

We call a \emph{state variable assignment}, denoted $\alpha$,
an assignment with domain $V$. In contrast, a \emph{transition variable assignment}, denoted $\beta$, is a (partial)
function that assigns values to the annotated variables $V^r\cup V^w$,
used to specify how variables change as the result of activity executions (cf. Def.~\ref{def:tfiring}).
Again, we require that $\beta(x)\in \Dom(\sort(x))$, for $x\in V^r\cup V^w$.

% sarah added to addrss reviewer comment
For a DPN $\NN$ with underlying Petri net $(P, T, F, \ell)$, a \emph{marking} $M:P\mapsto\mathbb{N}$ assigns every place a number of tokens.
A \emph{state} of $\NN$ is a pair $(M,\alpha)$ of a marking and a state variable assignment, which thus accounts for both the control flow progress and the current values of variables in $V$. %, as specified by $\alpha$. 
For instance,
$(\{\m p_0\}, \domino{t}{0}{o}{0})$ is a state for the net of \exaref{auction}.
We next define when transitions may fire in a DPN. 

\begin{definition}[Transition firing]
\label{def:tfiring}
A transition $t\in T$ is \emph{enabled} in a state $(M, \alpha)$ if a transition variable assignment $\beta$ exists such that:
\begin{compactenum}[\it (i)]
\item  $\beta(v^r) = \alpha(v)$ for every $v\in \vread(t)$, i.e., 
$\beta$ assigns read variables as by $\alpha$,
% $\beta$ is as $\alpha$ on read variables;
\item  $\beta \models \guard(t)$, i.e., $\beta$ satisfies the guard; and 
\item  $M(p) \geq F(p,t)$ for every $p$ so that $F(p,t)\geq 0$. % $p \in \pre{t}$.
\end{compactenum}
An enabled transition may \emph{fire}, producing a new state $(M',\alpha')$, s.t. $M'(p) = M(p) - F(p,t) + F(t,p)$ for every $p\in P$, and $\alpha'(v) \eqn \beta(v^w)$ for every $v\in \vwrite(t)$, and $\alpha'(v) \eqn \alpha(v)$ for every $v\not\in \vwrite(t)$. 
A pair $(t,\beta)$ as above is called (valid) \emph{transition firing}, and we denote its firing by $\smash{(M, \alpha) \goto{(t, \beta)} (M', \alpha')}$.
\end{definition}

\noindent
Given $\NN$, we fix one state $(M_I,\alpha_0)$ as \emph{initial}, where $M_I$ is the initial marking of the underlying Petri net $(P, T, F, \ell)$ and $\alpha_0$ is a state variable assignment that specifies the initial value of all variables in $V$.  Similarly, we denote the final marking as $M_F$, and call \emph{final} any state of the form $(M_F,\alpha_F)$ for some $\alpha_F$.
For instance, the net in \exaref{auction} admits a transition firing 
$\smash{(\{\m p_0\}, \domino{t}{0}{o}{0}) \goto{\m{init}}
(\{\m p_1, \m p_2\}, \domino{t}{1}{o}{0})}$ from its initial state, while
$(\{\m p_3\}, \domino{t}{0}{o}{5})$ is one final state.

We say that $(M',\alpha')$ is \emph{reachable} in a DPN iff there exists a sequence of transition firings
$(M_I,\alpha_0)\goto{(t_1, \beta_1)} \ldots\goto{(t_n, \beta_n)}(M',\alpha')$,  denoted also as $(M_I, \alpha_0) \to^* (M',\alpha')$.
Such a sequence is a (valid) \emph{process run} if the resulting state $(M',\alpha')$ is final. 
For instance, a possible sequence of transition firings in \exaref{auction} (in which the timer $t$ is initialized to $1$ day, then decremented) is: 
%has the following sequence of transition firings:
\begin{equation}
\label{eq:auctiondeadlock}
(\{\m p_0\}, \domino{t}{0}{o}{0}) \goto{\m{init}}
(\{\m p_1, \m p_2\}, \domino{t}{1}{o}{0}) \goto{\m{timer}}
(\{\m p_1, \m p_2\}, \domino{t}{0}{o}{0}) 
\end{equation}
% We say that $(M',\alpha')$ is \emph{reachable} in a DPN iff there exists a sequence of transition firings
% $\procrun =(t_1, \beta_1), \dots, (t_n, \beta_n)$, s.t. $(M_I,\alpha_0)\goto{(t_1, \beta_1)} \ldots\goto{(t_n, \beta_n)}(M',\alpha')$,  denoted as $(M_I, \alpha_0) \goto{\procrun} (M_n, \alpha_n)$. If we simply want to express that there exists such an $\procrun$, we write $(M_I, \alpha_0) \to^* (M_n, \alpha_n)$.
% Moreover, $\procrun$ is called a (valid) \emph{process run} of $\NN$ if $(M_I,\alpha_0)\goto{\procrun} (M_F,\alpha_F)$ for some $\alpha_F$. 
%
For simplicity of presentation, in the remainder of this paper, we restrict to \emph{bounded} DPNs, that is, DPNs where the number of tokens in reachable markings is bounded by some $m\inn\mathbb N$. 
Indeed, note that detecting unboundedness (which, in turn, witnesses unsoundness) can be done analogously to \cite{deLeoniFM21}. There,  it is shown that the standard unboundedness detection techniques based on coverability graphs seamlessly apply to the data-aware setting.
%
%Indeed, unbounded DPNs are known to violate condition (P2) in ~\defref{soundness}~\cite{deLeoniFM21}.
For instance, the DPNs $\NN$, $\NN_{\m{reset}}$, and $\NN_{\m{thresh}}$ in \exaref{auction} are 1-bounded.
Next, we define the crucial property of data-aware soundness.

\begin{definition}[Data-aware soundness]
\label{def:soundness}
A DPN is \emph{data-aware sound} iff:
\begin{compactitem}
\item[\textup{\textbf{(P1)}}] if $(M_I, \alpha_0) \to^* (M, \alpha)$ there is some $\alpha'$ such that $(M, \alpha) \to^* (M_F, \alpha')$ for all $M$, $\alpha$, i.e., any sequence can be continued to a process run;
\item[\textup{\textbf{(P2)}}] if $(M_I, \alpha_0) \to^* (M, \alpha)$ and $M\,{\geq}\,M_F$ then $M \eqn M_F$ for all $M$, $\alpha$, i.e., termination is clean; and
\item[\textup{\textbf{(P3)}}] for all $t \in T$ there 
is a sequence $(M_I, \alpha_0) \to^* (M, \alpha) \goto{(t, \beta)} (M', \alpha')$
for some $M$, $M'$, $\alpha$, $\alpha'$, and $\beta$, i.e., there are no dead transitions.
\end{compactitem}
\end{definition}

\noindent
For instance, the DPN $\NN$ from \exaref{auction} violates (P1) because after the sequence \eqref{auctiondeadlock} above
no further transition is applicable, but the reached state is not final.
$\NN_{\m{reset}}$ also violates (P3) because the transition \textsf{reset} is dead:
if a token reaches the place $\m p_3$, $o$ will never have value $0$.
On the other hand,
$\NN_{\m{thresh}}$ violates also (P2) as the following steps lead to marking $\{\m p_2, \m p_3\} > \{\m p_3\} = M_F$:
\begin{equation*}
\label{eq:exdirtytermination}
(\{\m p_0\}, \domino{t}{0}{o}{0}) \goto{\m{init}}
(\{\m p_1, \m p_2\}, \domino{t}{1}{o}{0}) \goto{\m{bid}}
(\{\m p_1, \m p_2\}, \domino{t}{1}{o}{1000})  \goto{\m{thresh}}
(\{\m p_2, \m p_3\}, \domino{t}{1}{o}{1000}) 
\end{equation*}

\newcommand{\hasDirtyTermination}{\textsc{badTermination}\xspace}
\newcommand{\hasDeadTransition}{\textsc{deadTransition}\xspace}
\newcommand{\hasBlockedState}{\textsc{blockedState}\xspace}
\newcommand{\computeConstraintGraph}{\textsc{computeCG}\xspace}
\newcommand{\DPNtoDDS}{\textsc{DPNtoDDS}\xspace}

\section{Soundness Checking: The High-Level Perspective}
\label{sec:overview}

\algref{soundness} gives a bird's-eye view of our soundness checking procedure.
The initial step is to transform the given DPN $\NN$ into a special kind of transition system (called DDS) $\BB$, by unfolding the interleaving semantics. The respective procedure \DPNtoDDS is detailed in \secref{dds}.
Next, in line 3, the procedure \computeConstraintGraph constructs the constraint graph of $\BB$ as a symbolic representation of all reachable states, as explained in \secref{cg}. In lines 4, 6, and 8 the routines \hasDirtyTermination, \hasDeadTransition, and \hasBlockedState then use the constraint graph $\CG$ to check whether $\NN$ violates the properties (P2), (P3), and (P1) of \defref{soundness}, respectively (see \secref{soundness}). If one of these properties does not hold, the procedure returns $\mathit{false}$ immediately, otherwise data-aware soundness is confirmed by returning $\mathit{true}$ in line 10. The reason why we check (P1) last is that the other two checks are significantly cheaper.
\begin{algorithm}
\caption{Procedure to check data-aware soundness of a DPN}\label{alg:soundness}
\begin{algorithmic}[1]
\Procedure{checkSound}{$\NN$}
  \State $\BB\gets \DPNtoDDS(\NN)$ 
  \State $\CG\gets \computeConstraintGraph(\BB)$
  \State\textbf{if} $\hasDirtyTermination(\CG, \NN)$ \textbf{then} \textbf{return} $\mathit{false}$
  \Comment{see \algref{properties}}
  %\State{\quad\textbf{return} $\mathit{false}$}
  \State\textbf{if} $\hasDeadTransition(\CG, \NN)$ \textbf{then} \textbf{return} $\mathit{false}$
  \Comment{see \algref{properties}}
  %\State{\quad\textbf{return} $\mathit{false}$}
  \State\textbf{if} $\hasBlockedState(\CG, \NN)$ \textbf{then} \textbf{return} $\mathit{false}$
  \Comment{see \algref{properties}}
  %\State{\quad\textbf{return} $\mathit{false}$}
  \State \textbf{return} $\mathit{true}$
\EndProcedure
\end{algorithmic}
\end{algorithm}

\section{From DPNs to Transition Systems}
\label{sec:dds}

This section details the first step in our soundness checking procedure: to unfold the interleaving semantics of the given DPN into a labelled transition system called \emph{data-aware dynamic system} (DDS)~\cite{LFM20}.
% In this section, we give the respective definition, and then show how a DPN can be transformed into a DDS by the procedure $\DPNtoDDS$ (cf. \algref{soundness}).
We start by defining DDSs.

\begin{definition}
\label{def:DDS}
A \emph{DDS} $\BB=\langle B, \binit, \AA, \Delta, B_F, V, \alphainit, \guard\rangle$ is
a labelled transition system such that
\begin{inparaenum}[(i)]
\item $B$ is a finite set of \emph{states}, with $\binit\inn B$ the initial one;
\item $\AA$ is a set of \emph{actions};
\item $\Delta \subseteq B \times \AA \times B$ is a \emph{transition relation};
\item $B_F \subseteq B$ are \emph{final states};
\item $V$ is the set of \emph{process variables};
\item $\alphainit$ is the \emph{initial assignment};
\item $\guard\colon \AA \mapsto \CC(V^r \cup V^w)$ fixes \emph{executability constraints} 
on actions. % over variables $V^r\,{\cup}\,V^w$.
\end{inparaenum}
\end{definition}

\noindent
\figref{ddss} shows three example DDSs (that are in fact obtained from transforming the DPNs in \exaref{auction}, as defined below). 
The action guards are the same as in \figref{dpn}, but have been omitted for readability.
We denote a transition from state $b$ to $b'$ by executing an action $a\inn \AA$ as 
$b \goto{a} b'$.
For instance, the DDS $\BB$ in \figref{ddss} admits a transition 
$\m p_0 \goto{\m{init}} \m p_{12}$.
A \emph{configuration} of $\BB$ is a pair $(b, \alpha)$ where $b\inn B$
and $\alpha$ is an assignment.
For instance, $(\m p_0, \domino{t}{0}{o}{0})$ is the initial configuration of $\BB$ in \figref{ddss}.
An \emph{action firing} is a pair $(a, \beta)$ of an action $a \inn \AA$ and a
transition variable assignment $\beta$, i.e., a function $\beta\colon V^r \cup V^w \mapsto D$. As defined next, an action firing $(a, \beta)$ transforms a configuration $(b, \alpha)$ into a new configuration $(b', \alpha')$ by changing state as defined by  action $a$, and updating the assignment $\alpha$ to $\alpha'$, in agreement with the action guard. In the new assignment $\alpha'$, variables that are not written keep their previous value as per $\alpha$, whereas written variables are updated according to $\beta$.
Let $write(a) = \{x \mid x^w\in V^w\text{ occurs in }\guard(a)\}$.

\begin{definition}
A DDS $\BB\,{=}\,\langle B, \binit, \AA, \Delta, B_F, V, \alphainit, \guard\rangle$
\emph{admits a step} from configuration $(b, \alpha)$ to 
$(b', \alpha')$ via action firing $(a, \beta)$,
denoted $(b, \alpha) \goto{a, \beta} (b', \alpha')$,
if $b \goto{a} b'$ and
\begin{inparaenum}[(i)]
\item $\beta(v^r) = \alpha(v)$ for all $v \in V$;
\item the new state variable assignment $\alpha'$ satisfies $\alpha'(v) = \alpha(v)$
if $v \in V \setminus write(a)$, and
$\alpha'(v) = \beta(v^w)$ otherwise; 
\item $\beta \models \guard(a)$, i.e., the guard is satisfied by $\beta$.
\end{inparaenum}
\end{definition}

\noindent
Thus, the variable update works exactly as for the case of DPNs.
For instance, $\BB$ in \figref{ddss} admits a step 
$(\m p_0, \domino{t}{0}{o}{0}) \goto{\m{init},\beta}
(\m p_{12}, \domino{t}{1}{o}{0})$ where $\beta(t^r) = \beta(o^r) = \beta(o^w) = 0$ and
$\beta(o^w) = 1$.
Given a DDS $\BB$, a \emph{derivation} $\rho$ of length $n$ from a configuration $(b, \alpha)$ is a sequence of steps:
\[\smash{\rho\colon(b, \alpha) = (b_0, \alpha_0) 
\goto{a_1, \beta_1} (b_1, \alpha_1)
\goto{a_2, \beta_2} \dots
\goto{a_n, \beta_n} (b_n, \alpha_n)}\]
We also associate with $\rho$ the \emph{symbolic derivation} $\sigma$ that \emph{abstracts} $\rho$, i.e., the sequence
$\smash{\sigma\colon b_0 \goto{a_1} b_1 \goto{a_2} \dots \goto{a_n} b_n}$
where only the state and action sequences are recorded, but no concrete assignments
are given. For some $m<n$, 
$\sigma|_{m}$ is the prefix of $\sigma$ that has $m$ steps.
We call a \emph{run} of $\BB$ a derivation starting from $(\binit, \alphainit)$, and a \emph{symbolic run} a symbolic derivation starting from $\binit$.
For instance,
\begin{equation}
\label{eq:deadlockrun}
\rho\colon(\m p_0, \domino{t}{0}{o}{0}) \goto{\m{init}}
(\m p_{12}, \domino{t}{1}{o}{0}) \goto{\m{timer}}
(\m p_{12}, \domino{t}{0}{o}{0}) 
\end{equation}
is a derivation of the DDS $\BB$ from \figref{ddss}, and also a run because it starts in the initial state $\m p_0$; $\rho$ is abstracted by the symbolic run
$\m p_0 \goto{\m{init}} \m p_{12} \goto{\m{timer}} \m p_{12}$. One may notice the similarity with the sequence of transition firings \eqref{auctiondeadlock} in \secref{background}.
\smallskip

\noindent
\textbf{Transformation.}
It is straightforward to define the procedure $\DPNtoDDS(\NN)$ used in \algref{soundness} to transform a given, bounded DPN $\NN$ into a DDS.
To this end, we consider in the rest of this section a $k$-bounded DPN $\NN = \langle P, T, F, \ell, \AA, V, \guard\rangle$ with initial variable assignment $\alphainit$, initial marking $M_I$, and final marking $M_F$. 
We define $\DPNtoDDS(\NN)$ as the DDS 
$\BB = \langle B, M_I, \AA, \Delta, \{M_F\}, V, \alphainit, \guard\rangle$ where $B$
is the set of all $k$-bounded markings of $\NN$; and $(M, a, M') \in \Delta$ iff
there is some $t\inn T$ such that $\ell(t) = a$,
$M(p) \geq F(p,t)$ for every %$p \in \pre{t}$, 
$p$ so that $F(p,t)\geq 0$ 
and
$M'(p) = M(p) - F(p,t) + F(t,p)$ for every $p\in P$.
Indeed, \figref{ddss} shows the DDSs obtained for the DPNs $\NN$, $\NN_{\m{reset}}$, 
and $\NN_{\m{thresh}}$ from \exaref{auction}.

After having defined the transformation from DPNs to DDSs, it remains to relate
data-aware soundness of a DPN with properties of its DDS representation. 
To that end, we define some notions that turn out to be useful:
%
% \begin{definition}
A DDS $\BB\,{=}\,\langle B, \binit, \AA, \Delta, B_F, V, \alphainit, \guard\rangle$
has a \emph{blocked state} if there is a run
$\smash{\rho\colon(\binit, \alphainit) \to^* (b, \alpha)}$
to some configuration $(b, \alpha)$ such that there is no derivation
$(b, \alpha) \to^* (b_f, \alpha')$ with $b_f\inn B_F$.
% \end{definition}
% \noindent
Moreover, let a state $b\in B$ be \emph{reachable} if there is a run 
$(\binit, \alphainit) \to^* (b, \alpha)$ for some $\alpha$; 
and a transition $(b, a, b')\in \Delta$ be \emph{reachable} if there is a run 
$(\binit, \alphainit) \to^* (b, \alpha) \goto{a, \beta} (b', \alpha')$ for some $\alpha$, $\alpha'$, and $\beta$.
It is then not hard to observe the following relationship between 
the properties (P1), (P2), and (P3) in \defref{soundness} and
properties of the DDS representation:

\begin{lemma}
\label{lem:dpndds}
If $\NN$ is a DPN and $\BB = \DPNtoDDS(\NN)$ has control states $B$,
\begin{compactitem}
\item $\NN$ satisfies (P1) iff $\BB$ has no blocked states,
\item $\NN$ satisfies (P2) iff all $M \inn B$ with $M \geq M_F$ are unreachable, and
\item $\NN$ satisfies (P3) iff for all transitions $t \inn T$ of $\NN$
there are some $M, M'\in B$ such that $(M, \ell(t), M') \in \Delta$ is reachable.
\end{compactitem}
\end{lemma}

\begin{figure}[t]
\centering
\begin{tikzpicture}[node distance=14mm,>=stealth']
\node[place, minimum width=7mm] (p0) {$p_0$};
\node[place, below of=p0, minimum width=7mm] (p12) {$p_{12}$};
\node[place, below of=p12,double, minimum width=7mm] (p3) {$p_3$};
\draw[edge] (p0) -- node[right,action,anchor=west] {$\mathsf{init}$}(p12);
\draw[edge, loop left] (p12) to 
 node[action,anchor=east,yshift=3mm, xshift=4mm, near end] {$\mathsf{bid}$} (p12);
\draw[edge, loop right] (p12) to
 node[action,anchor=west,yshift=2mm] {$\mathsf{timer}$} (p12);
\draw[edge] (p12) -- node[right, action,anchor=west] {$\mathsf{hammer}$} (p3);
\end{tikzpicture}
\qquad
\begin{tikzpicture}[node distance=14mm,>=stealth']
\node[place, minimum width=7mm] (p0) {$p_0$};
\node[place, below of=p0, minimum width=7mm] (p12) {$p_{12}$};
\node[place, below of=p12,double, minimum width=7mm] (p3) {$p_3$};
\draw[edge] (p0) -- node[right,action,anchor=west] {$\mathsf{init}$}(p12);
\draw[edge, loop left] (p12) to 
 node[action,anchor=east,yshift=3mm, xshift=4mm, near end] {$\mathsf{bid}$}(p12);
\draw[edge, loop right] (p12) to
 node[action,anchor=west,yshift=2mm] {$\mathsf{timer}$} (p12);
\draw[edge] (p12) -- node[right, action,anchor=west] {$\mathsf{hammer}$} (p3);
\draw[edge, rounded corners] (p3) -- node[action,below] {$\mathsf{reset}$}  ($(p3) + (1.6,0)$) -- ($(p0) + (1.6,0)$) -- (p0);
\end{tikzpicture}
\qquad
\begin{tikzpicture}[node distance=14mm,>=stealth']
\node[place, minimum width=7mm] (p0) {$p_0$};
\node[place, below of=p0, minimum width=7mm] (p12) {$p_{12}$};
\node[place, below of=p12,double, minimum width=7mm, xshift=-7mm] (p3) {$p_3$};
\node[place, below of=p12, minimum width=7mm, xshift=7mm] (p23) {$p_{23}$};
\draw[edge] (p0) -- node[right,action,anchor=west] {$\mathsf{init}$}(p12);
\draw[edge, loop left] (p12) to 
 node[action,anchor=east,yshift=3mm, xshift=4mm, near end] {$\mathsf{bid}$}(p12);
\draw[edge, loop right] (p12) to
 node[action,anchor=west,yshift=2mm] {$\mathsf{timer}$} (p12);
\draw[edge] (p12) -- node[left, action,anchor=east] {$\mathsf{hammer}$} (p3);
\draw[edge] (p12) -- node[right, action,anchor=west] {$\mathsf{thresh}$} (p23);
\draw[edge, loop right] (p23) to 
 node[action,yshift=3mm,near end] {$\mathsf{timer}$}(p23);
\end{tikzpicture}
\caption{DDSs $\BB$, $\BB_{\m{reset}}$, and $\BB_{\m{thresh}}$ for DPNs $\NN$, $\NN_{\m{reset}}$, and $\NN_{\m{thresh}}$.}
\label{fig:ddss}
\end{figure}
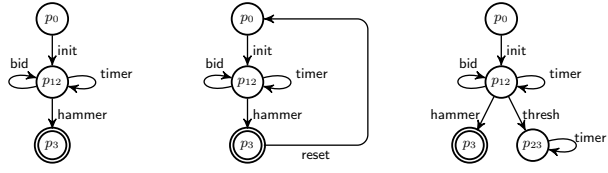

\noindent
This relationship allows us to  check data-aware soundness on the level of DDSs.
For instance, $\m{reset}$ is not reachable in $\BB_{\m{reset}}$ as $\m p_3$ is only reached via $\m{hammer}$, i.e., if $o > 0$,
so $\NN_{\m{reset}}$ does not satisfy (P3).
Also, $\BB_{\m{thresh}}$ admits the run (\ref{eq:threshrun}) below to state $p_{23}$, corresponding to marking $\{p_2,p_3\}$ in $\NN_{\m{thresh}}$,  
%so $\NN_{\m{thresh}}$ does not satisfy (P2). 
violating (P2). 
Finally, $\BB$, $\BB_{\m{thresh}}$ and $\BB_{\m{reset}}$ have the blocked state $(\{\m p_1, \m p_2\}, \domino{t}{0}{o}{0})$, reachable
via run \eqref{deadlockrun}. 
\begin{equation}
\label{eq:threshrun}
\rho\colon(\m p_0, \domino{t}{0}{o}{0}) \goto{\m{init}}
(\m p_{12}, \domino{t}{1}{o}{0}) \goto{\m{bid}}
(\m p_{12}, \domino{t}{1}{o}{1001}) \goto{\m{thresh}}
(\m p_{23}, \domino{t}{1}{o}{1001})
\end{equation}

\section{Constraint Graph}
\label{sec:cg}

While numerical data and arithmetic are required to faithfully model processes in many real-life information systems, they render the state space infinite.
For instance, the DDS $\BB$ in \figref{ddss} has infinitely many configurations such as 
$\smash{(\m p_{12}, \domino{t}{1}{o}{5})}$, $\smash{(\m p_{12}, \domino{t}{2}{o}{3})}$,
and $\smash{(\m p_{12}, \domino{t}{0}{o}{3})}$.
However, not all state variable assignments differ with
respect to possible next actions:
action \textsf{hammer} requires $o\,{>}\,0$ and $t \leq 0$, while \textsf{bid} and \textsf{timer} need $t\,{>}\,0$; 
but it is irrelevant whether, say, $o\,{>}\,4$.
Therefore, $\smash{(\m p_{12}, \domino{t}{1}{o}{5})}$ and $\smash{(\m p_{12}, \domino{t}{2}{o}{3})}$ 
are indeed \emph{equivalent} with respect to possible next steps, but the configurations
$\smash{(\m p_{12}, \domino{t}{2}{o}{3})}$ and 
$\smash{(\m p_{12}, \domino{t}{0}{o}{3})}$ are not.
Now, the key idea of the constraint graph is to symbolically represent
equivalent configurations using a tuple
$(b,\varphi)$ of a control state $b$ and a formula $\varphi$ over variables $V$.
For instance, for $\BB$ we will distinguish $(\m p_{12}, (o \eqn 0) \wedge (t\,{>}\,0))$ (both \textsf{bid} and \textsf{timer} apply) from $(\m p_{12}, (o \eqn 0))$ (we have no information about $t$, so only \textsf{bid} applies).

To formalize this idea, let 
$\BB\,{=}\,\langle B, \binit, \AA, \Delta, B_F, V, \alphainit, \guard\rangle$ be a given DDS.
We start with some auxiliary notions:
The \emph{transition formula} $\trans{a}$ of action $a$ is given by
$\trans{a}(\vec V^r, \vec V^w)\:{=}\:
\guard(a) \wedge \bigwedge_{v\not\in \vwrite(a)} v^{w}\,{=}\,v^{r}$. 
It simply expresses conditions on variables \emph{before and af\-ter} executing the action: $\guard(a)$ must hold, and the values of all variables that are not written are copied. 
E.g., for action $\m{bid}$ in \figref{ddss}, we have
$\vwrite(\m{bid}) = \{o\}$, and
$\trans{\m{bid}} = (t^r\,{>}\,0) \wedge (o^w\,{>}\,o^r) \wedge (t^{w}\,{=}\,t^{r})$.
Next, we use the transition formula to define an \emph{update} operation, representing how a current state, captured by a formula $\phi$, changes when executing action $a$.

\begin{definition}
\label{def:update}
For a formula $\phi$ %with free variables $V \cup V_0$ 
and action $a$, let 
$\update(\phi, a) = \qe(\exists \vec U. \phi[\vec U/\vec V] \wedge \Delta_a[\vec U/\vec V^r, \vec V/\vec V^w])$,
where $U$ is a set of variables that has the same cardinality as $V$ and is 
disjoint from all variables in $\phi$.
\end{definition}

\noindent
Here, $\phi[\vec U/\vec V]$ is the result of replacing variables $\vec V$
in $\phi$ by $\vec U$, and similar for $\trans{a}$.
For instance, if $\vec V = (o,t)$ we can take the renamed variables $\vec U = (o',t')$;
for $\phi = (t\,{>}\,0) \wedge (o \eqn 0)$ we then get
$\update(\phi, \m{bid}) = \qe(\exists o'\,t'.\: (t'\,{>}\,0) \wedge (o' \eqn 0) \wedge (o\,{>}\,o') \wedge (t\,{=}\,t'))$, which is simplified by quantifier elimination 
to $(t\,{>}\,0) \wedge (o\,{>}\, 0)$.
The use of a quantifier in \defref{update} might look like a complication, but it allows us to remember the previous state $\phi$, even if variables are overwritten by action $a$; afterwards,
quantifier elimination can produce a logically equivalent formula without $\exists$.
Next, given assignment $\alpha$, let $C_\alpha$ be the formula $C_\alpha\doteq\bigwedge_{v\in V} v = \alpha(v)$.

\begin{definition}\label{def:cg}
A \emph{constraint graph} $\CGof{b_0,\alpha}$ for $\BB$, a state $b_0\inn B$,
and assignment $\alpha$ is a triple 
$\langle S, s_0, \gamma\rangle$ where the set of nodes $S$
consists of tuples $(b, \phi)$ for $b \inn B$
and a formula $\phi$, %with free variables $V \cup V_0$, 
and $\gamma \subseteq S \times \AA \times S$,
inductively defined as follows:
\begin{compactitem}
\item[$(i)$] $s_0 = (b_0, C_{\alpha}) \in S$ is the initial node; and
\item[$(ii)$] if $(b,\phi) \in S$ and $b \goto{a} b'$ such that
$\update(\phi, a)$ is satisfiable, there is some $(b', \phi')\in S$
with $\phi' \equiv \update(\phi, a)$, and 
$(b,\phi) \goto{a} (b', \phi')$ is in $\gamma$.
\end{compactitem}
\end{definition}

\noindent
Intuitively, the constraint graph describes symbolically the states reachable in $\BB$.
% and in \algref{soundness} we give the procedure $\computeConstraintGraph(\BB)$ which computes it. 
Specifically, we write $\CG$ for the graph $\CGof{\binit, \alphainit}$ starting at the initial state and the initial assignment. This is also the graph returned by the procedure
$\computeConstraintGraph(\BB)$ used in \algref{soundness}.
For instance, the first two graphs in \figref{cgs} show $\CG$ and $\CG[\BB_{\m{thresh}}]$, respectively. 
Nodes that have the control state $\m p_3$ that is final in $\BB$ are drawn with double border; the coloring will be explained later. 

For technical reasons, our procedure often requires to consider constraint graphs that are built from an arbitrary state $b$ and that, instead of assigning variables $V$ to specific values, only impose that they have the same value of fresh placeholder variables $V_0$. 
We denote this by $\CGof{b}$. %\todo{simplified without $\alpharen$}%%PPP%%$\CGren{b}$. 
E.g., the rightmost graph in \figref{cgs} shows $\CGof{\m p_{12}}$, 
representing the states reachable in $\BB$ from $\m p_{12}$ where $V = \langle o,t\rangle$ is initially assigned the placeholder variables $V_0= \langle o_0,t_0\rangle$.

\begin{figure}[t]
\tikzstyle{node} = [draw,rectangle split, rectangle split parts=2,rectangle split horizontal, rectangle split draw splits=true, inner sep=3pt, scale=.65, rounded corners=2pt]
\tikzstyle{goto} = [->]
\tikzstyle{highlight} = [red!80!black, line width=.8pt]
\tikzstyle{action}=[scale=.6, black]
\tikzstyle{final}=[double]
\tikzstyle{dirty}=[fill=red!80!black!20]
\begin{tikzpicture}
\node[node] (0)  {\cgnode{$\m p_0$}{$o\,{=}\,0 \wedge t\,{=}\,0$}};
\node[node, below of=0] (1)
 {\cgnode{$\m p_{12}$}{$o\,{=}\,0 \wedge t\,{>}\,0$}};
\node[node, below of=1, xshift=10mm, highlight] (2)
 {\cgnode{$\m p_{12}$}{$o\,{=}\,0$}};
\node[node, below of=2, xshift=-10mm] (3)
 {\cgnode{$\m p_{12}$}{$o\,{>}\,0 \wedge t\,{>}\,0$}};
\node[node, below of=3] (4)
 {\cgnode{$\m p_{12}$}{$o\,{>}\,0$}};
\node[node, below of=4,final] (5)
 {\cgnode{$\m p_{3}$}{$o\,{>}\,0\wedge t\,{\geq}\,0$}};
\draw[goto] (0) to node[action, right]{$\m{init}$} (1);
\draw[goto] (1) to node[action, right, xshift=2mm]{$\m{timer}$} (2);
\draw[goto, bend right] (1) to  node[action, left]{$\m{bid}$} (3);
\draw[goto] (2) to  node[action, right]{$\m{bid}$} (3);
\draw[goto, loop right, out=7, in=-7, looseness=8] (3) to  node[action, right]{$\m{bid}$} (3);
\draw[goto, bend right] (3) to  node[action, left]{$\m{timer}$} (4);
\draw[goto, bend right] (4) to  node[action, right]{$\m{bid}$} (3);
\draw[goto] (4) to  node[action, right]{$\m{hammer}$} (5);
\end{tikzpicture}
\begin{tikzpicture}
\node[node] (0)  {\cgnode{$\m p_0$}{$o\,{=}\,0 \wedge t\,{=}\,0$}};
\node[node, below of=0] (1)
 {\cgnode{$\m p_{12}$}{$o\,{=}\,0 \wedge t\,{>}\,0$}};
\node[node, below of=1, xshift=10mm] (2)
 {\cgnode{$\m p_{12}$}{$o\,{=}\,0$}};
\node[node, below of=2, xshift=-10mm] (3)
 {\cgnode{$\m p_{12}$}{$o\,{>}\,0 \wedge t\,{>}\,0$}};
\node[node, below of=3, xshift=-13mm] (4)
 {\cgnode{$\m p_{12}$}{$o\,{>}\,0$}};
\node[node, below of=4, xshift=-3mm,final] (5)
 {\cgnode{$\m p_{3}$}{$o\,{>}\,0\wedge t\,{\geq}\,0$}};
\node[node, below of=3, xshift=13mm,dirty] (7)
 {\cgnode{$\m p_{23}$}{$o\,{>}\,1000 \wedge t\,{>}\, 0$}};
\node[node, below of=7,dirty] (6)
 {\cgnode{$\m p_{23}$}{$o\,{>}\,1000$}};
\draw[goto,highlight] (0) to node[action, right,highlight]{$\m{init}$} (1);
\draw[goto] (1) to node[action, right, xshift=2mm]{$\m{timer}$} (2);
\draw[goto, bend right,highlight] (1) to  node[action, left,highlight]{$\m{bid}$} (3);
\draw[goto] (2) to  node[action, right]{$\m{bid}$} (3);
\draw[goto, loop right, out=6, in=-6, looseness=6] (3) to  node[action, right]{$\m{bid}$} (3);
\draw[goto, bend right] (3) to  node[action, left]{$\m{timer}$} (4);
\draw[goto, bend right] ($(4.north) + (.3,0)$) to  node[action, right]{$\m{bid}$} (3.200);
\draw[goto] (4) to  node[action, left]{$\m{hammer}$} (5);
\draw[goto] (4) to  node[action, right, xshift=2mm]{$\m{thresh}$} (6.north west);
\draw[goto] (7) to  node[action, right, xshift=2mm]{$\m{dec}$} (6.north);
\draw[goto,highlight] (3) to  node[action, right, xshift=2mm,highlight]{$\m{thresh}$} (7);
\draw[goto, loop right, out=6, in=-6, looseness=6, near start] (6) to  node[action, above, xshift=1mm]{$\m{timer}$} (6);
\end{tikzpicture}
\begin{tikzpicture}[node distance=9mm]
\node[node] (0)  {\cgnode{$\m p_{12}$}{$o\,{=}\,o_0 \wedge t\,{=}\,t_0$}};
\node[node, below of=0, xshift=-10mm, yshift=-24mm] (1)
 {\cgnode{$\m p_{12}$}{$o\,{>}\,o_0 \wedge t\,{=}\,t_0 \wedge t\,{>}\,0$}};
\node[node, below of=0, yshift=-6mm, xshift=-7mm] (2)
 {\cgnode{$\m p_{12}$}{$o\,{=}\,o_0 \wedge t\,{<}\,t_0\wedge t_0\,{>}\,0$}};
\node[node, below of=0, xshift=14mm, yshift=2mm, final] (3)
 {\cgnode{$\m p_{3}$}{$o\,{=}\,o_0\,{>}\,0 \wedge t\,{=}\,t_0\,{>}\,0$}};
\node[node, below of=2, final, xshift=5mm] (4)
 {\cgnode{$\m p_{3}$}{$o\,{=}\,o_0\,{>}\,0 \wedge t\,{<}\,t_0\wedge t\,{\leq}\,0\wedge t_0\,{>}\,0$}};
\node[node, below of=1] (5)
 {\cgnode{$\m p_{12}$}{$o\,{>}\,o_0 \wedge t_0\,{>}\,0\wedge t_0\,{>}\,t$}};
\node[node, below of=5, xshift=5mm] (6)
 {\cgnode{$\m p_{12}$}{$o\,{>}\,o_0 \wedge t_0\,{>}\,t\,{>}\,0$}};
\node[node, below of=6, final, yshift=1mm] (7)
 {\cgnode{$\m p_{3}$}{$o\,{>}\,o_0 \wedge o \,{>}\,0 \wedge t_0\,{>}\,t\wedge t\,{\leq}\,0\wedge t_0\,{>}\,0$}};
\draw[goto] (0.west) .. controls (-2.4, -.9) and (-2.7, -1.8) .. node[action, left, very near start, yshift=1mm]{$\m{bid}$} (1.north west);
\draw[goto] (0.200) to node[action, left]{$\m{timer}$} (2);
\draw[goto] (0.south east) to  node[action, right, yshift=1mm]{$\m{hammer}$} (3.70);
\draw[goto] (2) to node[action, right]{$\m{hammer}$} (4);
\draw[goto, loop left, out=175, in=185, looseness=5] (1) to  node[action, below, near end]{$\m{bid}$} (1);
\draw[goto, loop right, out=5, in=-5, looseness=5] (2) to  node[action, right]{$\m{timer}$} (2);
\draw[goto] (1) to node[action, left]{$\m{timer}$} (5);
\draw[goto, loop right, out=5, in=-5, looseness=5] (6) to  node[action, right]{$\m{bid}$} (6);
\draw[goto, loop left, out=175, in=185, looseness=5] (5) to  node[action, below, near end]{$\m{timer}$} (5);
\draw[goto, bend right] (5) to node[action, left]{$\m{bid}$} (6);
\draw[goto, bend right] (6) to node[action, right]{$\m{timer}$} (5);
\draw[goto, bend right] ($(5.south west) + (.2,0)$) to node[action, left]{$\m{hammer}$} ($(7.north west) + (.8,0)$);
\end{tikzpicture}
\caption{Constraint graphs $\CG$, $\CG[\BB_{\m{thresh}}]$, and $\CGof{\m p_{12}}$.\label{fig:cgs}}
\end{figure}
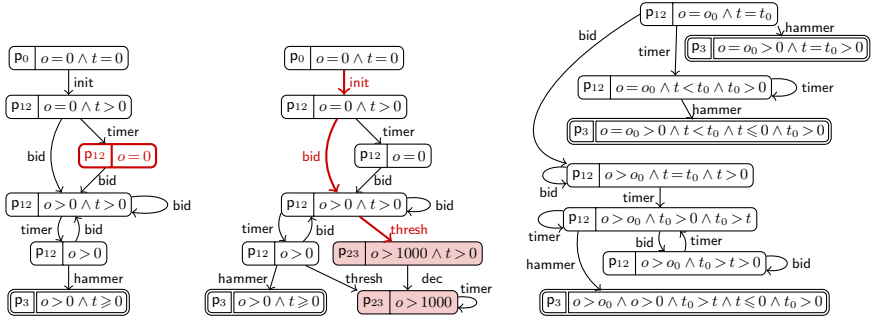

We next establish properties that connect constraint graphs to derivations
of the DDS $\BB$. 
For a \emph{path}  
$\pi \colon (b_0,\phi_0) \goto{a_1} (b_1,\phi_1) \goto{a_2} \dots \goto{a_n} (b_n, \phi_n)$
in a constraint graph, we denote by $\sigma(\pi)$ the symbolic derivation 
$b_0 \goto{a_1} b_1 \goto{a_2} \dots \goto{a_n} b_n$ that has the same control state and action sequences.
We now show that every combination of a path in a constraint graph and a satisfying assignment for the formula in its final node corresponds to a run in the DDS, and vice versa. To that end, we need a fixed variable renaming $\hat \alpha\colon V \mapsto V_0$.

\begin{lemma}
\label{lem:cg}
\begin{compactenum}
\item[(1)]
$\CG$ has a path $\pi\colon (\binit, C_{\alphainit}) \to^* (b,\phi)$ where
$\phi$ is satisfiable by $\alpha$,
iff $\BB$ has a run
$\smash{(\binit, \alphainit) \to^* (b, \alpha)}$
whose abstraction is $\sigma(\pi)$. 
\item[(2)]
%%PPP%%$\CGren{b}$ 
$\CGof{b}$ has a path $\pi\colon (b,C_{\alpharen}) \to^* (b',\phi)$ s.t.
$\phi$ is satisfiable by $\alpha$,
iff $\BB$ has derivation
$\smash{(b, \alpha_0) \to^* (b', \alpha_n)}$
abstracted by $\sigma(\pi)$ with $\alpha_0\eqn\alpha|_{V_0}$ and $\alpha_n\eqn\alpha|_V$.
\end{compactenum}
\end{lemma}
\begin{proof}
We sketch the proof of (1), the case for (2) is very similar.
For the direction from left to right, we apply induction on $\pi$. 
If $\pi$ is empty, it must end in node $(\binit, C_{\alphainit})$. 
As $C_{\alphainit}$ is only satisfied by $\alphainit$, and $(\binit, \alphainit)$
is a valid (empty) run, the claim holds.
For the inductive step, consider a path
$\pi\colon (\binit, C_{\alphainit}) \to^* (b,\phi) \goto{a} (b', \phi')$ of length $n+1$.
Assume $\alpha$ satisfies $\phi' \equiv \update(\phi,a)$. Since $\qe$ preserves equivalence, $\alpha$ satisfies $\exists \vec U. \psi$ where $\psi \doteq \phi[\vec U/\vec V] \wedge \Delta_a[\vec U/\vec V^r, \vec V/\vec V^w]$. Thus, there must be an extension 
$\alpha'$ of $\alpha$ with domain $V \cup U$ that satisfies $\psi$. Let $\alpha_n$ be such that $\alpha_n(\vec V) = \alpha'(\vec U)$, so $\alpha_n \models \phi$, and
 $\beta$ be the transition variable assignment such that $\beta(\vec V^r) = \alpha'(\vec U)$ and $\beta(\vec V^w) = \alpha'(\vec V)$, so $\beta \models \trans{a}$.
By the induction hypothesis, $\BB$ has a run
$\smash{(\binit, \alphainit) \to^* (b', \alpha_n)}$
whose abstraction is $\sigma(\pi)|_{n}$, and by definition of the substitutions
we can extend it with the step
$\smash{(b, \alpha_n) \goto{a,\beta} (b', \alpha)}$.

The direction from right to left is even easier, we reason by induction on
the given run. In the inductive step, we again apply the definition of $\update$.
\qed
\end{proof}

\noindent
To illustrate this result, e.g. the run \eqref{threshrun} corresponds to the path in $\CG[\BB_{\m{thresh}}]$ shown in red (see \figref{ddss}). On the other hand, the lemma reveals that $\BB$ has runs with the same action sequence for \emph{all} assignments that satisfy $(o\,{>}\,1000)\wedge (t\,{>}\,0)$.
%\smallskip

As stated above, the construction of the constraint graph according to \defref{cg} need not terminate. However, it does in many practical examples, which is related to the following property identified in~\cite{ada}: 
A DDS $\BB$ has a \emph{finite history set} if the set of
formulas $\phi$ obtained during the construction of the constraint graph (called \emph{history constraints} in \cite{ada}) is finite up to equivalence.
Thus, if $\BB$ has a finite history set, and the procedure $\computeConstraintGraph(\BB)$ checks eagerly for equivalent nodes
while executing \defref{cg}, the construction must produce
a finite graph.
Crucially, this holds for a clearly identifiable class of systems used in the literature: it was shown that if either the constraint language in $\BB$ is restricted to variable-to-variable and variable-to-constant comparisons,
or if the control flow is such that the current state depends only on finitely many actions in the past, the DDS $\BB$ has indeed a finite history set~\cite[Thms 5.2 and 5.9]{ada}. All examples of DPNs collected from the literature (see \secref{implementation}) fall in one of these categories.

\section{Data-aware Soundness}
\label{sec:soundness}

In this section we harness the constraint graph to check data-aware soundness.
To that end, we assume a
DDS $\BB = \langle B, b_I, \AA, \Delta, \{b_F\}, V, \alphainit, \guard\rangle$
obtained by translating a DPN $\NN$,
such that $b_I = M_I$ and $b_F = M_F$ correspond to the initial and final markings of $\NN$; and we assume that $\CG$ is the constraint graph of $\BB$. 
The three requirements of \defref{soundness} are then checked by the procedures in \algref{properties}:

\begin{algorithm}[t]
\caption{Checking soundness properties for $\CG = \langle S, s_0, \gamma\rangle$ and
DPN $\NN$}\label{alg:properties}
\begin{algorithmic}
% \Require{$\CG = \langle S, s_0, \gamma\rangle$, and
% DPN $\NN$ with final marking $M_F$ and transitions $T$}
\\\vspace{-1ex}
\Procedure{\hasDirtyTermination}{$\CG, \NN$}
  \State{\textbf{return} $\exists (b,\phi)\in S$ such that $b$ corresponds to marking $M$ in $\NN$ and $M\,{\geq}\,M_F$}
\EndProcedure
\\\vspace{-2.5ex}
\Procedure{\hasDeadTransition}{$\CG, \NN$}
  \State{\textbf{return} $\exists$ transition $t$ of $\NN$ such that $\nexists (s, \ell(t), s') \in \gamma$ for some $s$, $s'$}
\EndProcedure
\\\vspace{-2.5ex}
\Procedure{\hasBlockedState}{$\CG, \NN$}
  \State{\textbf{return} $\exists (b,\phi)\in S$ such that $b\neq b_F$ and $\mathit{blocked}(b,\phi)$ satisfiable}
\EndProcedure
\end{algorithmic}
\end{algorithm}

$\bullet$ \hasDirtyTermination returns \emph{true} if in the node set
$S$ of the constraint graph $\CG$ there is a node $(b,\phi)$ such that $b$
corresponds to a marking $M$ of the DPN $\NN$ with $M \geq M_F$.
For instance, it returns \emph{true} for $\CG[\BB_{\m{thresh}}]$
in \figref{cgs}
since the red nodes correspond to marking $\{\m p_2, \m p_3\}$; while it would return \emph{false} for $\CG[\BB]$. 

$\bullet$ \hasDeadTransition returns \emph{true} if there is a transition in the DPN $\NN$ whose label does not occur in $\CG$.
For instance, the constraint graph for the DDS $\BB_{\m{reset}}$ in \figref{ddss}
coincides with the graph $\CG[\BB]$ in \figref{cgs}, which does not contain \textsf{reset}.
Thus, $\hasDeadTransition(\CG,\NN_{\m{reset}})$ returns \emph{true}.

$\bullet$ For \hasBlockedState, we use the formulas $\mathit{blocked}(b,\phi)$  defined next.
For $b\inn B$ and constraint graph
$\CGof{b} 
= \langle S', \gamma', s_0'\rangle$, let 
$\mathit{final}(b) \eqn \{\phi \mid (b_F, \phi) \inn S'\}$ 
be all formulas in $\CGof{b}$ that occur together with final states.
Then,

\begin{definition}
For $\CG = (S, \gamma, s_0)$ and $(b,\phi) \in S$, let
\begin{small}
\[
\label{eq:deadlock}
\textstyle
\mathit{blocked}(b,\phi) = 
\phi[\vec V/\vec V_0] \wedge 
\neg \left (\exists \vec V. \bigvee_{\psi \in \mathit{final}(b)}\: \psi \right ).
\]
\end{small}
\end{definition}

\noindent
Now, \hasBlockedState returns \emph{true} if there is some 
node $(b,\phi)\in S$ in $\CG$ such that $\mathit{blocked}(b,\phi)$ is satisfiable.
This formula basically expresses that the process reaches control state $b$ with an assignment from which the final state is not reachable: Indeed,
$\Psi := \exists \vec V. \bigvee_{\psi \in \mathit{final}(b)} \psi$ expresses conditions
to reach a final state from $b$ and variables assigned to $\vec V_0$ (the existential quantifier reflects that we do not care
about the final values of the data variables).
Thus, $\neg \Psi$ states that \emph{no} final state can be reached, and we take the 
conjunction with $\phi$ (with variables renamed appropriately) to combine this with the assumptions of the current constraint graph node $(b,\phi)$.
For instance, 
we can check whether $\BB$ in \figref{ddss} admits a deadlock at a run the is captured by the 
node $(\m p_{12}, o\eqn 0)$ (drawn in red) of $\CG$ in \figref{cgs}, as follows:
There are three final nodes in $\CGof{\m p_{12}}$ in \figref{cgs} labelled
$\phi_1 \doteq (o \eqn o_0 \wedge o_0\,{>}\,0 \wedge t \eqn t_0 \wedge t_0\,{>}\,0)$,
$\phi_2 \doteq (o \eqn o_0 \wedge o_0\,{>}\,0 \wedge t_0\,{>}\,t \wedge t\,{\leq}\,0\wedge t_0\,{>}\,0)$, and
$\phi_3 \doteq (o\,{>}\,o_0 \wedge o\,{>}\,0 \wedge t_0\,{>}\,t \wedge t\,{\leq}\,0\wedge t_0\,{>}\,0)$, so $\mathit{final}(\m p_{12}) = \{\phi_1, \phi_2, \phi_3\}$.
We hence get
\begin{small}
\[
\mathit{blocked}(\m p_{12}, o\eqn 0) = (o_0 \eqn 0) \wedge 
\neg \left (\exists o\: t.\: (\phi_1 \vee \phi_2 \vee \phi_3 )\right )
\]
\end{small}
% 
% for the node $(\m p_{12}, o=0)$ (drawn in red) of $\CG$ in \figref{cgs},
% and using the final nodes in $\CGof{\m p_{12}}$ in \figref{cgs} to obtain $\mathit{final}(\m p_{12})$, 
% we get the formula
% \begin{small}
% \begin{align*}
% \mathit{blocked}(\m p_{12}, o\eqn 0) = (o_0 \eqn 0) \wedge 
% \neg \left (\exists o\: t. \right.
% &(o \eqn o_0 \wedge o_0\,{>}\,0 \wedge t \eqn t_0 \wedge t_0\,{>}\,0) \vee {} \\[-.5ex]
% &(o \eqn o_0 \wedge o_0\,{>}\,0 \wedge t_0\,{>}\,t \wedge t\,{\leq}\,0\wedge t_0\,{>}\,0) \vee {} \\[-.5ex]
% &(o\,{>}\,o_0 \wedge o\,{>}\,0 \wedge t_0\,{>}\,t \wedge t\,{\leq}\,0\wedge t_0\,{>}\,0) 
% \left.\right )\left.\right )
% \end{align*}
% \end{small}
% \noindent
which is simplified using quantifier elimination to
$(o_0 \eqn 0) \wedge (t_0\,{\leq}\,0)$, and e.g.
satisfiable by $\alpha(o_0) \eqn \alpha(t_0) \eqn 0$.
Thus $\hasBlockedState(\CG, \NN)$ returns \emph{true},
reflecting the blocked sequence \eqref{auctiondeadlock} shown at the end of \secref{background}. 

Note that all checks in \algref{properties} are effective if $\CG$ and all $\CGof{b}$ are finite.
Finally, we relate the procedures in \algref{properties} to properties of $\BB$,
which together with \lemref{dpndds} shows that data-aware soundness of DPNs is effectively checked.

\begin{theorem}
\label{thm:blocked}
Let $\CG$  be a constraint graph for a DDS $\BB$.
\begin{compactenum}[(1)]
\item $\hasBlockedState(\CG,\BB)$
returns \textit{true} iff $\BB$ has a blocked state.
\item $\hasDeadTransition(\CG,\BB)$
returns \textit{true} iff $\NN$ has a transition $t$ of $\NN$
such that $(b, \ell(t), b') \in \Delta$ is unreachable for all $b, b'\in B$, and
\item $\hasDirtyTermination(\CG,\BB)$
returns \textit{true} iff some $b\in B$ corresponding to $M$ with 
$M \geq M_F$ is reachable.
\end{compactenum}
\end{theorem}
\begin{proof}
For (1)
$(\Longleftarrow)$, if $\BB$ has a blocked state, there is a run
$\rho\colon (\binit, \alphainit) \to^* (b, \alpha)$
such that $b \neq b_F$
but no continuation
$(b, \alpha) \to^* (b_F, \alpha_f)$ to the final state.
Let $\sigma$ be the abstraction of $\rho$.
By \lemref{cg} (1), there is a node $(b,\phi)$
in $\CG$ reachable by a path $\pi$ with $\sigma(\pi) = \sigma$ such that $\alpha \models \phi$.
We show that the assignment $\alpha'$ such that $\alpha'(\vec V_0) = \alpha(\vec V)$ satisfies
$\mathit{blocked}(b,\phi)$, so that $\hasBlockedState(\CG,\BB)$ returns $\mathit{true}$.
For the sake of a contradiction, suppose there is some $\psi \in \mathit{final}(b)$
such that $\alpha' \models \exists \vec V.\: \psi$,
i.e., there is an extension $\alpha''$ of $\alpha'$ with domain $V_0 \cup V$ such that $\alpha''\models \psi$.
By definition of $\mathit{final}$, $(b_F, \psi)$ is a node in $\CGof{b}$.
However, according to \lemref{cg} (2), there is a
derivation $(b, \alpha''|_{V_0}) \to^*(b_F, \alpha''|_V)$. This contradicts that no final
configuration is reachable from $(b, \alpha) = (b, \alpha''|_{V_0})$.

\noindent
For (1) $(\Longrightarrow)$, suppose $\hasBlockedState(\CG,\BB)$ returns $\mathit{true}$, so there is a node
$(b,\phi)$ in $\CG$, reachable by some path $\pi$, such that the formula
$\mathit{blocked}(b,\phi)$ is satisfied by some assignment $\alpha$.
By \lemref{cg} (1), $\sigma(\pi)$
abstracts a run $\rho\colon(\binit,\alphainit) \to^* (b,\alpha)$.
For the sake of a contradiction, suppose there is a derivation 
$\rho'\colon(b,\alpha) \to^* (b_f,\alpha_f)$ such that $b_f\in B_F$. 
Let $\alpha'$ be the assignment with domain $V \cup V_0$ such that
$\alpha'(\vec V_0) = \alpha(\vec V)$ and $\alpha'(\vec V) = \alpha_f(\vec V)$.
For the abstraction $\sigma'$ of $\rho'$, by \lemref{cg} (2) the constraint graph
$\CGof{b}$ has a path $\pi'$ to a node $(b', \psi)$ such that $\sigma' = \sigma(\pi')$ 
and $\alpha' \models \psi$.
Since $\alpha'(\vec V_0) = \alpha(\vec V)$, it must hold that
$\alpha \models \exists \vec V.\: \psi$.
However, as $\psi \in \mathit{final}(b)$ this contradicts that $\alpha$ satisfies $\mathit{blocked}(b,\phi)$.
For (2) and (3), it is clear from \lemref{cg} (1) that $\CG$ contains a transition $t$, or control state $b$, iff $\BB$ has a run in which $t$ or $b$ occurs. 
\qed
\end{proof}

\section{Implementation and Experiments}
\label{sec:implementation}

We implemented our approach in the tool \tool (arithmetic DDS analyzer) in Python;
source code, benchmarks, and a web interface are available.\footnote{\url{https://soundness.adatool.dev}}
The tool takes a (bounded) DPN in {\tt .pnml} format %~\cite{pnml} 
as input, and checks data-aware soundness following  \algref{soundness} and \algref{properties}.
As output, it produces graphical representations of the DDS $\BB$ and the constraint graph $\CG$, and if data-aware soundness is violated, a witness is constructed.
% , i.e., a run of $\BB$ which ends in a blocked state; a reachable marking that properly contains a final marking; or an unused transition.
% 
For SMT checks and quantifier elimination, \tool interfaces CVC5~\cite{DR0BT14} and Z3~\cite{Z3}, which support all  datatypes mentioned in \secref{background}.

As DPNs are a relatively recent framework, an extensive set of benchmarks is
still missing. To mitigate this, we have collected all available DPN examples/use cases from the literature, and used \tool to check soundness. The results are shown in \tabref{tool}, which indicates data-aware soundness (and the violated property of \defref{soundness}),
the verification time, number of SMT checks, number of variables in the DDS $\BB$, and the sizes of $\BB$ and $\CG$ as number of nodes/transitions.
All tests were run on an Intel Core i7 ($4{\times}2.60$GHz, 19GB RAM), using CVC5 as backend. 

 \begin{table}[t]
\begin{center}
\scriptsize
\begin{tabular}{rll@{\,}|@{\,}c c@{\ }|r|r|r|r@{/}l|r@{/}l}
\multicolumn{3}{c|@{\,}}{\textbf{process}} & 
\multicolumn{2}{@{\,}c@{\,}|}{\textbf{sound}} & 
\multicolumn{1}{|c}{\textbf{time}} & 
\multicolumn{1}{|c|}{\textbf{checks}} &
$|V|$ & 
\multicolumn{2}{|c}{\textbf{$|\BB|$}} & 
\multicolumn{2}{|c}{\textbf{$|\CG|$}} \\
\hline
(1) & road fines (normative) & \cite[Fig. 7]{MannhardtLRA16}
&  no & P1 & 3.1s    & 3909    &  8 & 9&19   &   29&44   \\ \hline 
(2) & road fines (mined) &  \cite[Fig. 12.7]{Mannhardt18}
&  no & P1 & 3.1s  & 3811    &  8 & 9&19   &   59&104  \\ \hline 
(3) & road fines (mined) & \cite[Fig. 13]{LFM18}
& yes &      & 2m16s& 114,005 &    5 & 9&19   &  234&376  \\ \hline 
(4) & hospital billing & \cite[Fig. 15.3]{Mannhardt18}
& yes &      & 3m1s   & 229,467 &  4 & 17&40   &  360&703  \\ \hline 
(5) & sepsis (normative) & \cite[Fig. 13.3]{Mannhardt18} 
& yes &      & 19s   & 831   & 3  & 301&1630 &  793&4099 \\ \hline  
(6) & sepsis (mined) &  \cite[Fig. 13.6]{Mannhardt18}
& yes &      & 1m43s  & 8085   & 4 & 301&1630 & 1117&5339 \\ \hline  
(7) & digital whiteboard: register &  \cite[Fig. 14.3]{Mannhardt18}
& yes &      & 0.1s  & 16   & 2 & 7&6 & 7&6 \\ \hline  
(8) & digital whiteboard: transfer &  \cite[Fig. 14.3]{Mannhardt18}
& no & P1 & 0.1s  & 19    & 3 & 7&6 & 7&6 \\ \hline  
(9) & digital whiteboard: discharge  & \cite[Fig. 14.3]{Mannhardt18}
& yes &     & 0.1s  & 30   & 4 & 6&6 & 7&6 \\ \hline 
(10) & credit approval & \cite[Fig. 3]{LM18}
& yes &      & 1.2s  & 434   & 5  &   6&10   &   26&27   \\ \hline 
(11) & package handling & \cite[Fig. 5]{deLeoniFM21}      
& no  & P3 & 1.3s  & 242  & 5   &  16&28   &   68&67   \\ \hline 
% package handling (modified) 
% & yes & & 3.7s  & 768     &  16&20   &   68&67   \\ \hline 
(12) & auction & \cite[Ex. 1.1]{ada}          
& no  & P1 & 5.8s  & 1007   & 5 &   5&7    &   13&15\\
\hline
\end{tabular}
\end{center}
 \caption{Experiments with \tool on DPNs from the literature.\label{tab:tool}}
 \end{table} 
 
\noindent
We briefly comment on these benchmarks:
(1)--(3) model the handling of traffic offenses in an information system of the Italian police; in a normative model and two versions where decision rules were mined automatically from a log with 150k traces. The former two have the same unsoundness issue (see \exaref{road fines}), related
to missing guards on written variables. 
(4) models the billing process in a hospital, it was mined from a real-life log with 100k traces, discovering guards by overlapping decision mining.
(5) and (6) reflect the triage process for sepsis patients, based on a log obtained from a hospital's ERP system for 1,050 patients. (5) is a normative model,
whereas for (6), guards were discovered by decision mining.
(7)--(9) are activity patterns for patient logistics designed
based on domain knowledge and logs of a hospital information system.
(10) is a faithful though hand-made process of granting loans to clients of a bank.
(11) is a manually but realistically designed order-to-delivery process, obtained as a DPN translation of a DBPMN model (a data- and decision-aware model that builds on BPMN and DMN S-FEEL).
(12) is a manually designed model for an English auction.

We stress that the benchmarks (1), (5), (7), (10), and (12) are out of reach of
the earlier approaches~\cite{LFM18,deLeoniFM21}, as their constraint language cannot express addition and multiplication.
Moreover, while example (3) took
1.9h with the technique of~\cite{LFM18}, soundness can be detected by \tool in less than 3 minutes.

%As mentioned above, 
An extensive DPN benchmark set with a wide range of problem sizes is not yet available. To provide some indications on the scalability of our method, we therefore modified some of the above benchmarks, adding (a) up to 100 sequential control states, and (b) up to 10 data variables $z_1, \dots, z_k$ for every type, in the latter case obfuscating constraints of the form $e\odot e'$ to $e = z_1 \wedge z_1 = z_2 \wedge \dots \wedge z_k \odot e'$. The results are depicted in \figref{charts}, where the x-axis reports the number of added states/variables, and the y-axis the computation time.

 \begin{figure}[t]
\begin{tabular}{cc}
\begin{tikzpicture}
\begin{axis}[
xmin = 0, xmax = 95,
ymin = 0, ymax = 123,
width = 68mm,
height = 40mm,
legend cell align = {left},
legend pos = north west,
mark size=.8pt,
font=\tiny
]
\addplot[thin,red!80!black, mark=*] file[skip first] {data/road_fines_states.dat};
\addplot[thin,cyan!90!black, mark=*] file[skip first] {data/road_fines_m_states.dat};
\addplot[thin,blue!80!black, mark=*] file[skip first] {data/sepsis_states.dat};
\addplot[green!80!black, mark=*] file[skip first] {data/credit_approval_states.dat};
\addplot[thin,orange, mark=*] file[skip first] {data/package_handling_states.dat};
\end{axis}
% standard legend is ugly
\node[anchor = north west, draw, line width=.2pt, rectangle, scale=.55, fill=white] at (0.3, 2.2) {
\begin{tabular}{cl}
\color{red!80!black}\foo\color{black}&road fines (1) \\[-.3ex]
\color{cyan!90!black}\foo\color{black}&road fines (2) \\[-.3ex]
\color{blue!80!black}\foo\color{black}&sepsis (5) \\[-.3ex]
\color{green!80!black}\foo\color{black}&credit approval (10) \\[-.3ex]
\color{orange}\foo\color{black}&package handling (11)\\[-.3ex]
\end{tabular}
};
\end{tikzpicture}&
\begin{tikzpicture}
\begin{axis}[
xmin = 0, xmax = 10,
ymin = 0, ymax = 633,
width = 68mm,
height = 40mm,
legend cell align = {left},
legend pos = north west,
mark size=.8pt,
font=\tiny
]
% y = x1 = x2 = ... = xk < z
\addplot[thin,cyan!90!black, mark=*] file[skip first] {data/road_fines_m_vars2.dat};
\addplot[thin,red!80!black, mark=*] file[skip first] {data/road_fines_vars2.dat};
\addplot[thin,blue!80!black, mark=*] file[skip first] {data/sepsis_vars2.dat};
\addplot[green!80!black, mark=*] file[skip first] {data/credit_approval_vars2.dat};
\addplot[thin,orange, mark=*] file[skip first] {data/package_handling_vars2.dat};
% y < x1 < x2 < ... < xk < z
%\addplot[thin,cyan!90!black, dotted, mark=*] file[skip first] {data/road_fines_m_vars3.dat};
%\addplot[thin,red!80!black, dotted, mark=*] file[skip first] {data/road_fines_vars3.dat};
%\addplot[thin,blue!80!black, dotted,mark=*] file[skip first] {data/sepsis_vars3.dat};
%\addplot[green!80!black, dotted,mark=*] file[skip first] {data/credit_approval_vars3.dat};
%\addplot[thin,orange, dotted, mark=*] file[skip first] {data/package_handling_vars3.dat};
\end{axis}
\node[anchor = north west, draw, line width=.2pt, rectangle, scale=.55, fill=white] at (0.3, 2.2) {
\begin{tabular}{cl}
\color{red!80!black}\foo\color{black}&road fines (1) \\[-.3ex]
\color{cyan!90!black}\foo\color{black}&road fines (2) \\[-.3ex]
\color{blue!80!black}\foo\color{black}&sepsis (5) \\[-.3ex]
\color{green!80!black}\foo\color{black}&credit approval (10) \\[-.3ex]
\color{orange}\foo\color{black}&package handling (11)\\[-.3ex]
\end{tabular}
};
\end{tikzpicture}
\\
(a) scalability: sequential control-flow &
(b) scalability: data variables/constraints
\end{tabular}
\caption{Scalability of \tool considering control-flow (a) and data variables (b).}
\label{fig:charts}
\end{figure}
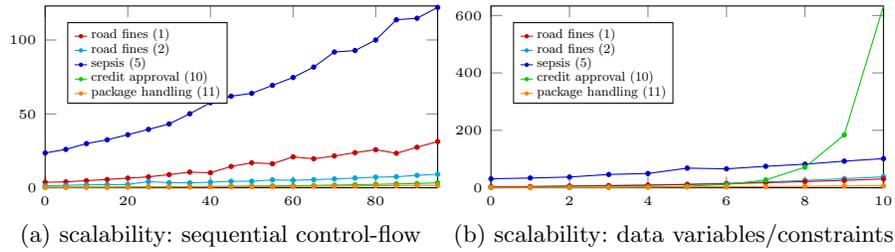

The chart in (a) suggests that the addition of sequential tasks in the control-flow increases the computation time only linearly. For (b), we also observe a linear behaviour for many systems; but for benchmarks with a more complex constraint structure such as the credit approval example, performance can be considerably harmed. However, note that the benchmarks generated in (b) exhibit far larger constraints than the real-world systems, and can hence be considered extreme cases. Finally, it it interesting to observe that similar trends are obtained for (b) when using operators other than equality in building the expanded constraints. 

% \paragraph{Optimizations.}
% \begin{itemize}
%  \item cache SMT formulas that were already checked for satisfiability
%  \item re-use already computed CGs: to get $\CGof{b}{\nu}$, if we already computed $\CGof{b'}{\nu}$ for some $b'$ such that
%  $b' \goto{a} b$ such that $\guard(a) = \top$, we obtain $\CGof{b}{\nu}$ by suitably restricting $\CGof{b'}{\nu}$.
%  On the other hand, if we already computed $\CGof{b'}{\nu}$ for some $b'$ such that
%  $b \goto{a} b'$ such that $\guard(a) = \top$ and this is the only step from $b$, we obtain $\CGof{b}{\nu}$ by adding a single edge to $\CGof{b'}{\nu}$.
%  \item skip some satisfiability checks of \eqref{deadlock}: if a node $s$ in $\CG(\NN)$ has only predecessors via actions with guard $\top$ for which \eqref{deadlock} was already checked, the check for $s$ can be skipped.
% \end{itemize}

\section{Conclusion}
\label{sec:conclusion}

%Many processes in real-world information systems require numerical data and
%arithmetic to be faithfully modeled; hence models such as DPNs emerged as a
%recent trend in BPM.
%On the other hand, increasingly complex processes became amenable to automatic 
%process discovery techniques, also in presence of data.
%In combination, these two developments render it highly intricate to manually check correctness properties such as data-aware soundness of the so-obtained process models. 
The presence of numerical data in data-aware process models, either designed by hand or discovered from logs, render it highly intricate (undecidable in general) to manually check correctness properties such as soundness. 
We presented the first automatic technique that can verify data-aware soundness for DPNs with linear arithmetic, along with a prototype implementation.
Our experiments show that the approach is effective and efficient, and can detect soundness bugs.

In future work, we aim at realizing a tighter integration between manual and automated approaches for data-aware process discovery and correctness analysis. Specifically, we plan to study the integration of this technique with automated approaches for process discovery to either guarantee by design the soundness of the discovered processes, or to provide specific indications on how to repair them %so as to guarantee correctness 
(e.g., by providing negative examples to be excluded or to guide the selection of fitness parameters when discovering decisions from those appearing in the log). We also intend to deepen our understanding of the scalability of the approach starting from the preliminary evaluation presented here, with the goal of isolating the main sources of computational complexity, and of incorporating specific methods to handle them.
Finally, we hope that having a solid foundational framework paired with a proof-of-concept IT artefact will trigger empirical research focussed on on-field validation of soundness for data-aware processes.

\bibliographystyle{splncs04}
% \bibliography{references}

\end{document}